\documentclass[]{article}

\usepackage{amsmath}
\usepackage{bbold}
\usepackage{graphicx}
\usepackage[english]{babel} 

\usepackage{placeins}
\usepackage[authoryear]{natbib}
\usepackage{amssymb}
\usepackage{titling} 
\usepackage{authblk}
\usepackage{multirow}

\usepackage{hyperref} 

\usepackage[margin=1in]{geometry}

\newcommand{\indep}{\perp \!\!\! \perp}

\usepackage{ntheorem}

\theoremstyle{plain}
\newtheorem{theorem}{Theorem}[section]

\newtheorem{lemma}[theorem]{Lemma}

\newtheorem{condition}{Condition}

\newtheorem{definition}{Definition}

\theorembodyfont{\normalfont}
\newtheorem{proofinner}{Proof}

\newenvironment{proof}[1][]
 {\if\relax\detokenize{#1}\relax
    \renewcommand\theproofinner{\thetheorem}%
  \else
    \renewcommand{\theproofinner}{#1}%
  \fi
  \proofinner}
 {\endproofinner}

\begin{document}
\title{Statistical Properties of Exclusive and Non-exclusive Online Randomized Experiments using Bucket Reuse}

\author[1,2]{M\aa rten Schultzberg\thanks{The authors thanks Andreas Born, Claire Detilleux, Brian St Thomas, Michael Stein, and the colleagues in the experimentation platform team for helpful feedback and suggestions for this paper.}}
\author[1]{Oskar Kjellin}
\author[1]{Johan Rydberg}
\affil[1]{Experimentation Platform Team – Spotify}
\affil[2]{Contact: mschultzberg@spotify.com}

\maketitle

\begin{abstract}
Randomized experiments is a key part of product development in the tech industry. It is often necessary to run programs of exclusive experiments, i.e., experiments that cannot be run on the same units during the same time. These programs implies restriction on the random sampling, as units that are currently in an experiment cannot be sampled into a new one. Moreover, to technically enable this type of coordination with large populations, the units in the population are often grouped into 'buckets' and sampling is then performed on the bucket level. This paper investigates some statistical implications of both the restricted sampling and the bucket-level sampling. 
The contribution of this paper is threefold: First, bucket sampling is connected to the existing literature on randomized experiments in complex sampling designs which enables establishing properties of the difference-in-means estimator of the average treatment effect. These properties are needed for inference to the population under random sampling of buckets. Second, the bias introduced by restricting the sampling as imposed by programs of exclusive experiments, is derived. Finally, simulation results supporting the theoretical findings are presented together with recommendations on how to empirically evaluate and handle this bias.
\end{abstract}

\pagebreak

\section{Introduction}
When building products and services in the internet age, it has become increasingly hard to figure out what is the 'right thing' to do. Companies like Microsoft, Netflix and Spotify serve a global audience where the 'right thing' may vary drastically between markets and demographics. To better understand what the customers need, and want, companies like these have started utilizing large scale randomized experiments. Over the last five to ten years, randomized experiments have become the de facto standard method for testing ideas. Randomized experiments, often referred to as A/B tests in this field, allow teams to test hypotheses in a scientific way, with proper risk management. This gives them the causal insight they need to systematically improve the user experience as new product changes are introduced. Today, these companies run several hundreds, or even thousands, of experiments every year  to learn about what their customers prefer. Even the smallest change can have a big impact \citep{Kohavi2017} -- and it is important that the experiment methods are efficient in terms of the amount of users needed to get a causal insight with sufficient precision. For example, Google famously tested 41 shades of blue to see which performed the best - the result increased revenue by \$200 millions \citep{hern2014}.

When experimentation scholars like Ronald A. Fisher, William Gosset ('Student'), Jerzy Neyman, and Egon S. Pearson \citep{fisher1935design, Student1908, Neyman1923, neyman1934} wrote their seminal papers on randomized experiments and inference to the units of the sample and the population, the concept of populations was central and well defined. W. Gosset (the man behind the t-test) was arguably the first product developer using randomized experiments. He used randomized experiments to optimize, among several things, the agriculture practice of Guinness (the beer manufacturer) to yield the best product in the late 1800 - early 1900. The population of fields was well known and Gosset could repeatedly (over years) sample and assign different plots to different fertilizers and seeds, to find the combinations that gave the best crops. When Neyman and Pearson generalized the math to infinite populations, it can be argued that 'populations' was more a mathematical construct that enabled deriving properties for more complicated estimators. In many fields where statistics has been applied, it has rarely been conceived as plausible that a population would ever be observed and explicitly available for repeated sampling and experimentation in practice -- especially in settings where human behaviours are studied. This is still the case in many areas of science. Often, so called, 'convenience samples' are used, still applying inference justified by random sampling to draw conclusions about some population parameter. Even in cases where the sample is random, it is rare with replications -- which of course weakens all arguments derived from repeated sampling. 

Interestingly, during the same period as many started to question statistical practices and frequentistic inference in the science community \citep{Ioannidis2005, Amrhein2019}, the tech industry rediscovered randomized experiments\citep{Kohavi2017}, inspired by the long tradition in the manufacturing industry of using so called 'a/b-tests' to optimize production. Within tech, incentives for bad practice like publication bias is less prominent \footnote{P-hacking still exists though, since it is easier to put a feature into production if there is a 'significant' experiment results to back it up.}, and it is easier to keep track of all experimentation results for meta analyses. The concepts of repeated sampling and randomized treatment assignment fits well into the ways of working in the tech industry. Quite ironically, in this setting where repeated sampling is a natural part of the practice, a new set of issues are encountered, issues that were not addressed by Fisher, Gosset, and Neyman. Perhaps a setting with data on hundreds of millions of users available with a click on a computer, was too inconceivable even for such progressive geniuses. The novel issues that are encountered are related to that these companies want to experiment more than they can. Having a fast build-measure-learn loop is key in modern product development. It is therefore critical to be able to to run experiments at any time. Even though the user bases for these companies are often in the hundreds of millions, the amount of available users for experimenting is often a limiting factor due to that the effect sizes of interest are small. This issue is amplified by the fact that, due to technical reasons, it is often required to run non-overlapping coordinated experiments. These requirements pose all sorts of new statistical questions: How can sampling and treatment assignment of so many users be randomized while making sure each user gets the correct 'treatment'? And more specifically, the questions that this paper addresses: How do technically plausible solutions to the sampling of users, in combination with programs of exclusive (non-overlapping) experiments, affect the statistical validity of the inference to the population? 

The rest of this paper is organized as follows. Section \ref{sec:coordexp} introduces the concept and challenges of coordinated experiments and the strategy that is the focus of this paper -- Bucket reuse. Section \ref{sec:number_of_buckets} investigates properties related to the number of buckets that the population is split into. Section \ref{sec:everchange} present statistical properties of the difference-in-means estimator of the average treatment effect under bucket reuse. Section \ref{sec:practical} discusses practical recommendation and Monte Carlo simulation results. Finally, Section \ref{sec:disc} gives concluding remarks.

\section{Coordinated Experiments}\label{sec:coordexp}
In most experimenting tech companies, 
many experiments are run in simultaneously.  One of the most critical tasks in the experimentation programs is to keep track of what user is in which experiments at any give time point. It is not only a matter of that one user can be in several experiments at once, but also that some users should not be allowed to be part of certain experiments simultaneously. For both technical and design reasons, it is not always possible to run several experiments on one user at one time. In the Spotify-app, e.g, one team might want to change the theme of the app to blue instead of green, while another team are trying to find the optimal green color. Of course, a user can only have one color per item at one time point. To be able to iterate on both these ideas at the same time, it is key to be able to experiment on both these ideas simultaneously. To achieve this, the experiments must be conducted on different users. In other words, the samples must be disjoint groups of users, and each sample only exposed to one experiment.  This is a well known challenge in experimenting tech-companies, see for example Google's solution in \cite{Tang2010}. 

Throughout this paper, we will focus on what we call \textit{programs of exclusive experiments}. A program of exclusive experiments is a set of experiments run over time where no users is exposed to more than one of the experiments in the program at a given time. To build intuition for the limitations imposed by running programs of exclusive experiments, it is helpful to introduce the concepts of paths. A path is simply a sequence of experiments that a unit can be in. If experiments are ran non-exclusively the number of paths is $2^{\text{(number of experiments)}}$, since each user either is or is not sampled to each experiment. In a program of exclusive experiments, the number of paths depends on how the program is run in terms of starts and stops of experiments. Figure \ref{fig:paths} illustrates a program of exclusive experiments containing 5 experiments. Below the experiments, the possible paths are displayed. For example, it is not possible to be sampled into both experiment 3 and 4 as they overlap in time and this is program of exclusive experiment. It is possible to be sampled experiment 1 and 4, since these are not overlapping in time.
\begin{figure}[hbt!]\centering
\caption{Paths of experiences possible during a program of exclusive experiments containing 5 exclusive experiments over time.}\label{fig:paths}
\includegraphics[scale=1.0, trim={5cm 11cm 3cm 4.3cm},page=8 , clip=TRUE]{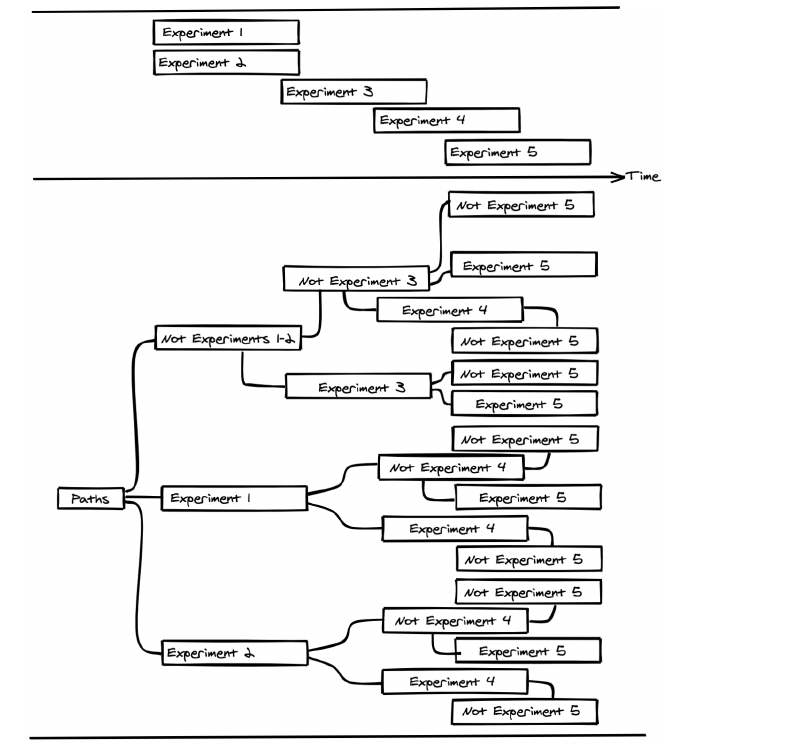}
\end{figure}
The number of unique possible paths explodes combinatorically after a relative short time period in most settings, and only a small partition of the possible paths can be taken by any units.

\subsection{The Technical Challenge of Coordination}
Designing a coordination tool that allows random sampling of users from the population, with and without restrictions to include or exclude users in certain other experiments, quickly becomes a technical challenge. For example, we have several hundred million users at Spotify, and at any given time point we run hundreds of experiments. Big data storage can easily store the assignment of each unit in each experiment. However, it is difficult to keep explicit records and look up who is in what experiment, and who should be allowed/not allowed to be sampled into a new experiment, in a manner that is sufficiently fast and reliable. At Spotify, for example, experimentation is a central part of all features and having any database in critical request paths puts experiments (or the feature) at risk of not getting the correct assignments if there would be some problems with the database or replication lag. Keeping all users assigned to all the experiments in memory becomes infeasible given the large number of users. The ability to keep all the ids of these users in memory depends on the number of bytes each id takes up, which may vary depending on what id the experiment targets on (user id being one out of multiple).

A common way of decreasing the memory load is to hash the users' id with a hash function \citep{knuth1998art} into a smaller set of groups of users. There exists multiple different types of hash functions, some are suited for experimentation whereas some are not. The primary attributes required of a hash function to be suitable for experimentation are: 1. Consistent for the same input. 2. No correlation between input and output (changing one bit of the input produces a completely new random output). 3. All outputs are equally likely (uniformity).
Given these properties, we can apply the hash function to the id which will give us bytes that are random (but consistent for the same id). We then take these bytes and turn them into an integer. This yields a simple way to consistently take an id and turn it into a random integer. To create meaningful groupings the unbounded output integer can be mapped into a fixed range of integers by applying a modulo operation. For example, we can transform any integer X into an integer between zero and hundred by applying $X \mod 100$. Altogether this gives us a way of uniformly spreading users based on their ids into any fixed number of groups, identified by the interger output of the modulo operation -- these groups are often called 'buckets'. Figure \ref{fig:hash} illustrates this mechanism. Hashing is a widely used method for handling very large groups of users in efficient manners.
\begin{figure}[h!]\centering
\caption{Illustration of the mapping between user id's and buckets using hashing with a random salt. }\label{fig:hash}
\includegraphics[scale=1., trim={5cm 14cm 5cm 4cm},page=9 , clip=TRUE]{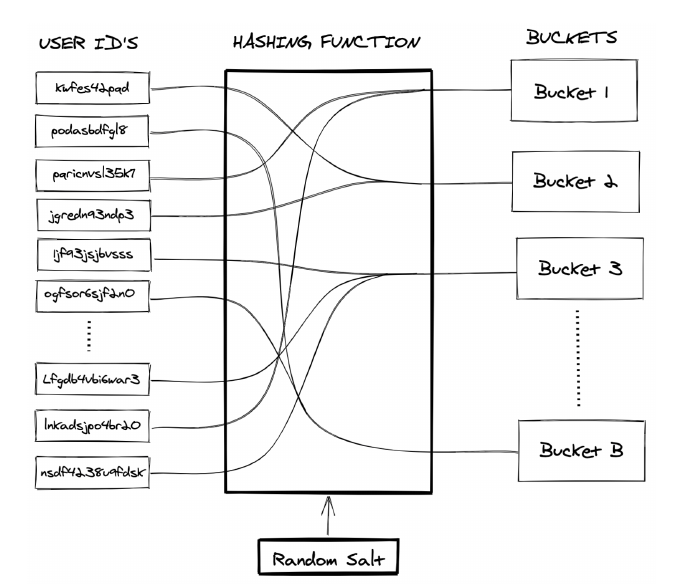}
\end{figure}

In experimentation, the buckets created by hashing can be used to efficiently sample users from the population. Given that the users are randomly assigned to buckets, we can easily assign 10\% of the buckets into one experiment and that experiment would on average get 10\% of the users. Assume we take the first 10\% of the buckets into the first experiment and continue taking contiguous sequences of buckets into the following 9 experiments. After doing this we have assigned all buckets into one experiment and those experiments will have no overlap in user ids (an id will only hash into 1 bucket ever). Now all users have seen exactly one experiment and which experiment that was is a deterministic function of what bucket the user hashed to.
If the initial buckets are randomly created, these 10 samples are random. However, once at least one experiment is stopped, and we want to start a new experiment – we would like to sample independently of the first round of experiment samples. A user may react differently if first exposed to experiment 1 and then to 3, vs exposed to 2 and then to 3. For this reason, it is desirable to randomly spread out all users from the 10 first experiments into all new experiments to average out this ordering effect. The most simple way to do this is to suffix the user that we hash with a random string (called a salt). We then suffix all the users with the same salt and when we want to rerandomize users into new buckets, we simply stop all experiments and introduce a new salt. Stopping all experiments running on A before introducing B is appealing, but it puts strong restrictions on the experimentation program. Assuming many teams are experimenting in the programs, as is often the case in the large tech-companies, a team that wants to run many short experiments will be blocked by teams that run longer experiments.

\subsection{Coordinated Experiments using Bucket Reuse}
A simpler solution than introducing new salts, as dicsussed above, is to use so called \textit{bucket reuse}. Bucket reuse consists of the following steps: Decide on a number of buckets ($B$). Take the unique user id and hash it together with a random salt into $B$ 'buckets'. To start an experiment, randomly sample buckets from the $B$ buckets, or what ever subset that is available at the time of sampling. As the name suggests, these buckets are \textit{reused} over time which means that no stopping is needed to start new experiments. For non-exclusive experiments this simply implies that the random sampling is done on the bucket level, rather than unit level. For experiments in programs exclusive of experiments, this implies sampling from the set of available buckets, i.e., the set of buckets that are not in an ongoing experiment at the time of sampling. Note that when we are talking about bucket reuse in this paper, we always mean sampling of buckets. Importantly, the buckets are not affecting the treatment assignment. That is, once the sample of buckets is obtained, the treatment assignment is randomly assigned on the unit level, independently of the buckets. 

Bucket reuse has been, and still is, used by several tech companies, but the statistical properties have not, to the best of our knowledge, been properly established. Clearly, using, and reusing, buckets will affect the statistical properties of common treatment effect estimators. It is important to remember that the reason for why large companies are investing so extensively into experimentation, is that the randomized sampling and treatment assignment is the best known tool for learning about causal effects in a statistically valid and robust way. If the efficient technical implementations put restrictions on the randomization, it is important to figure out what the statistical implications are, so that we do not build away the properties that we wanted in the first place, in our eagerness to solve the technical problem. The rest of this paper investigates the statistical properties under bucket reuse in non-exclusive experiments and programs of exclusive experiments.

\section{Choosing the Number of Buckets}\label{sec:number_of_buckets}
It is quite clear, even before the investigation, that the larger the number of buckets, the smaller the difference between bucket reuse and simple random sampling of units will be. If we have one bucket per unit, bucket sampling and unit sampling is equivalent. However, recall that the hashing was introduced because it is too hard to keep track of all units individually. For example, assume we have 1 billion units. At least one bit is required to represent a bucket. This means that 125 MB is needed to represent one experiment with one bucket per unit. 125 MB might not sound too bad, but this puts strict limitations on the number of concurrent experiments that can be run. If we want to run 1000 experiments concurrently, 125 GB is needed to represent this. As the number of concurrent experiments and number of units increase with time, this quickly becomes infeasible.  A virtual machine can easily get 125 GB of ram but at this point the cost starts to climb as well. Auto-scaling of virtual machines is a key part of being able to respond to traffic spikes. If each new machine that starts up needs to download 125 GB, the time it takes to scale up will take longer than wanted.

From a statistical perspective, the number of buckets used in bucket reuse changes the number of unique samples that can be drawn, and thereby the number of treatment-effect estimates that can be observed. As a small illustrative example, say that we have a population of 20 units, and that we want to draw a sample of 10 units. If we sample units individually, there are $\tbinom{20}{10}=184756$ unique samples. However, if we split the population into buckets of size 5, so that we have 4 distinct buckets, then the number of unique samples of size 10 is $\tbinom{4}{2}=6$. If we want to sample X number of users, we sample the number of buckets of which the sum of users comes closes to the desired number of users X. For example, with 1000 buckets the smallest proportion of the population that can be sampled is 0.1\%. Generally, any amount of the population that does not have modulo zero with $1/B$ cannot be sampled exactly.

When an experiment stops, we want the corresponding buckets to be spread into new experiments proportionally. That is, if we start X new experiments we want to randomly spread the buckets from the previously stopped experiment into the X new ones proportionally to the size of those new experiments. If a stopped experiment would only contain one bucket we could not spread the sample from this experiment across new experiments at all, instead that bucket would end up in one, or none, of the new experiments. If the number of buckets in an experiment is substantially larger than the number of started experiments, the buckets can be spread approximately according to the relative size of the new experiments. 
Consider the following example, displayed in Figure \ref{fig:spread}. At time point t-3 100\% of the population is used, and an experiment (A) is stopped freeing up 20\% of the population. At time point t-2 there is 20\% of the population available and no experiment is started or stopped. At time point t-1 one experiment (B) of 10\% is stopped and no experiment is started, implying 30\% available space in total. Say now that we want to start 2 experiments at time point t. The first experiment (C) will use 10\% of the population, and the second experiment (D) will use 5\% of the population. This implies that experiment C should in expectation contain 10/30=1/3 of the buckets in A and 1/3 of the buckets in B. Correspondingly, experiment D should  in expectation contain 5/30= 1/6 of the buckets in A and 1/6 of the buckets in B. For this to be possible, A and B must contain enough buckets for it to be possible to split them into 6. Of course, since we always have a finite number of buckets, there are many ratios of buckets that cannot be exactly achieved, but if the number of buckets is large, randomly selecting the buckets to C and D from all buckets from A and B will ensure the right amount on average with small deviations. 
\begin{figure}[ht!]\centering
\caption{Illustration of how bucket size restricts the randomization of a sample from an old experiment into a the sample of a new experiment. The dotted lines indicate hypothetical splits of a sample into smaller parts. }\label{fig:spread}
\includegraphics[scale=1., trim={5cm 15cm 3cm 4.5cm},page=10 , clip=TRUE]{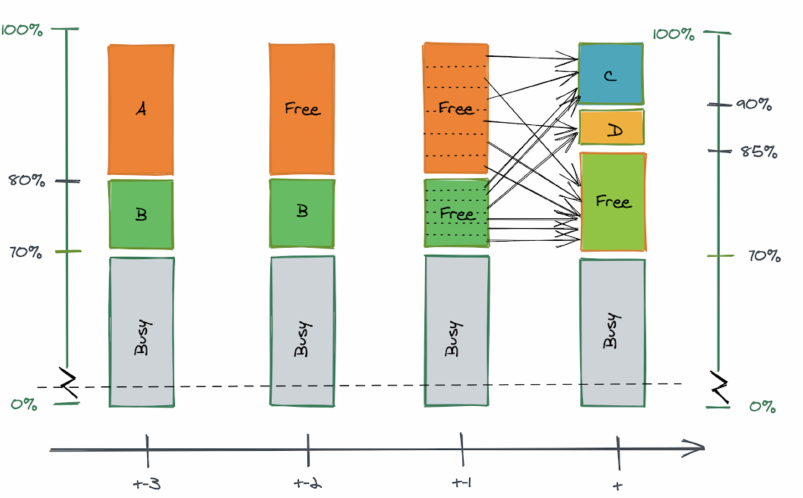}
\end{figure}
\FloatBarrier
The ability to proportionally split a sample that was in a previous experiment into the samples of new experiments is dependent on the size of the smallest experiment and the number of buckets (or equivalently, the size of the buckets). The smaller the smallest experiment is allowed to be, the more buckets we need to split our population into to be able to spread buckets from one of the stopped experiments of the smallest size proportionally into a new experiment of the smallest size. Figure \ref{fig:smallest_exp} displays this relation. For example, let the smallest allowed experiment be 0.1\% of the population. If we want to be able to stop an experiment (A) of size 0.1\%, and start a new one (B) of size 0.1\%. Then we want to randomly select 0.1\% of the 0.1\% buckets from the population that are in experiment A to put in the sample of experiment B. That is we want to be able to draw 0.1\% (1/1000) of the buckets in A – which of course implies that we have to have at least 1000 buckets in A. And if 0.1\% of the population is equal to 1000 buckets, then the total number of buckets needed is 1000/0.001=1000000 (1M). As we can see from the figure, the required number of total buckets increase very fast as the size of the smallest experiment approach 0\%. The corresponding required total number of buckets if the smallest experiment is allowed to be 0.05\% is 4M.
\begin{figure}[ht!]\centering
\caption{The relation between the smallest possible relative experiment sample size and number of required buckets for non-restricted randomization of previous samples into new samples.}\label{fig:smallest_exp}
\includegraphics[scale=0.9, trim={5cm 16cm 3cm 4cm},page=3 , clip=TRUE]{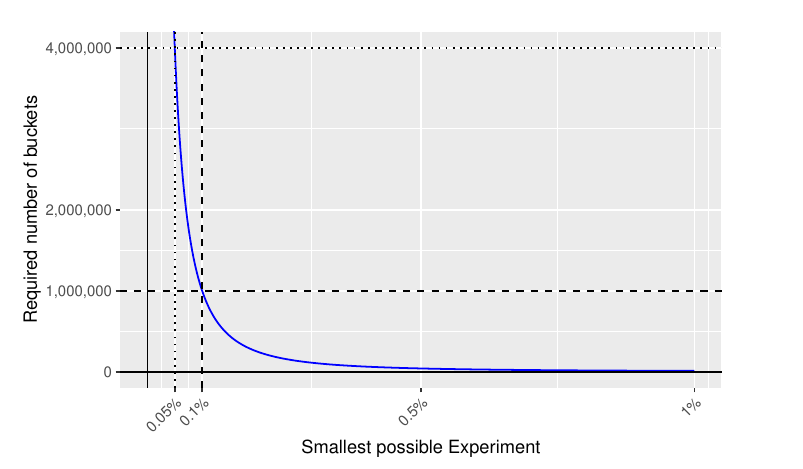}
\end{figure}
Note that the smallest allowed experiment size does not have to be enforced, but buckets from experiments that are smaller than that are not possible to include proportionally in new small experiments. In practice at Spotify, it is (so far) uncommon with very small experiments, and even more so with many such small experiments in a sequence.\\

Using bucket reuse, rather than sampling units directly, cannot increase the amount of units that go through bad experiences as long as an experiment sample is always larger than one bucket. However, the use of buckets does impact the probability of including disproportionally many units with previously strong experiences in one experiment. The following is an extreme example to illustrate this relation. Assume there are 0.1 percentage points (pp) units of the population with a neutral previous experience available for sampling, and that an experiment containing 0.1pp of units, that has had a 'bad' experience stops. That is, there is 0.2pp of the population available in total. Say we want to start an experiment of size 0.1pp, then the expected proportion of units with a previously neutral and bad experience in our random sample is 50\% each. The probability that the sample proportion is close to 50/50 neutral and bad in the sample depends on the number of buckets that makes up 0.1pp. 
Note that the probability of getting any specific proportion of 'bad buckets' can be exactly described using the hypergenometric distribution. Let $N_{bad}$ be the number of bad bucket in the sample and let $\rho_{bad}$ be the proportion of bad buckets in the sample, then
\begin{align*}
\rho_{bad}=\frac{N_{bad}}{\frac{N_S}{N_B}},
\end{align*}
and it follows that $Pr(\rho_{bad}=r) = Pr(N_{bad}=N_S*r),$ where 
$N_{bad}\sim \text{HypGeom}(B_{available}, N_{bad}, N_S/N_B)$, where $B_{available}$ is the number of available buckets to sample from.
If 0.1pp of the population is made up of one bucket, we will sample one of the two buckets, one is bad and one is neutral. The probability of getting a sample of 100\% bad buckets is 50\%. If 0.1\% is 2 buckets, then the probability of getting 100\% bad buckets in the sample is 16.7\%. Using the hypergeometric distribution, we can easily calculate the exact probabilities for getting different proportions of bad buckets in this example. The following figure displays the probability that the sample contains 50\%+- three different margins, 1pp, 2pp, and 3pp bad buckets.
\begin{figure}[ht!]\centering
\caption{The probability of obtaining 50\%$\pm$ an error margin of 'bad' buckets as a function of the number of buckets in the experiments. }\label{fig:hyper}
\includegraphics[scale=0.9, trim={5cm 16cm 3cm 4cm},page=4 , clip=TRUE]{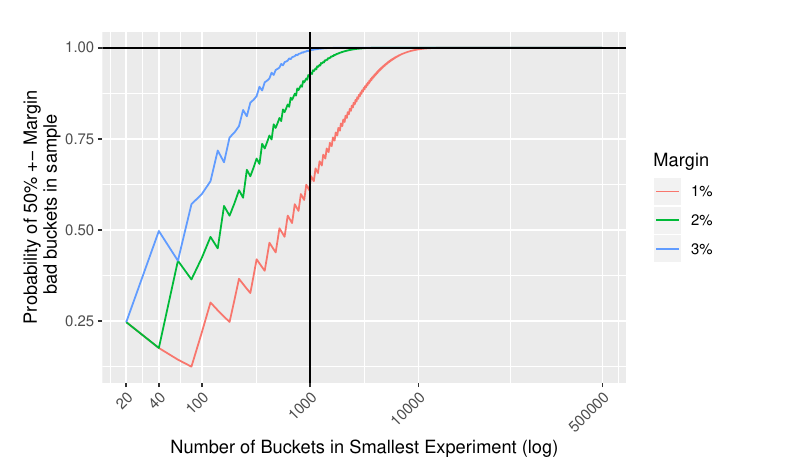}
\end{figure}
As we can see, the probability of sampling bad buckets within 1,2, and 3pp from 50 percent goes to one when the number of buckets increases. Clearly there is a difference in probability of having disproportionally many bad buckets if the smallest experiments contain 1000 bucket or 500k buckets. The probability of having more +-1pp too few/too many bad buckets is quite large with 1000 buckets. However, the probability of being more than +-3pp off is almost zero. When an experiment of size 0.1\% of the population contains 1000 buckets, the probability that the proportion of bad buckets is between 47\% and 53\% is 99.27\% in this scenario.

\subsection{Overlap across experiment programs}
If the same buckets are used in all programs of exclusive experiments and all non-exclusive experiment, the choice of the number of buckets needs to account for the relation between the sampling for experiments run within different programs or non-exclusive experiments. There is no dependency between the samples for experiments run in different programs or non-exclusive experiments, as the sampling is randomized within each program and for each non-exclusive experiment. This implies that the overlap of two experiments that are not in the same program is $\text{(proportion of population for experiment 1)}\times\text{(proportion of population for experiment 2)}$ in expectation. However, the number of buckets affects the probability distribution of the overlap and therefore the probability of the deviations from this expected value that will occur in practice. Again, the probability distributions can be described completely by Hypergenometric distributions. Assume, e.g., that a standard sample size for an experiment is 5\% of the population. Let experiments 1 and 2 be of size 5\% and run in different programs of exclusive experiments. Define $X=\text{The number of buckets from experiment 1 in experiment 2}$, it follows that $X\sim \text{HypGeom}(B, 0.05B, 0.05B)$. The expected overlap is $0.05^2=0.0025=0.25$ percentage points (pp) of the population. Table \ref{tab:hyp2} displays the probabilities that the sample in experiment 2 contains between 0.2pp and 0.03pp of the sample of experiment 1 as a function of number of buckets in the population.
\begin{table}[hbt!]\centering
\begin{tabular}{rrr}
  \hline
& \multicolumn{2}{c}{Experiment Size}\\
     \hline
 Number of buckets& 5pp & 10pp \\ 

  in population ($B$)  & \multicolumn{2}{c}{Overlap probability}\\
  \hline
1000 & 0.23 & 0.27 \\ 
   2000 & 0.34 & 0.37 \\ 
   10000 & 0.70 & 0.73 \\ 
  50000 & 0.98 & 0.99 \\ 
  100000 & 1.00 & 1.00 \\ 
   \hline
\end{tabular}
\caption{Probability of an overlap within 0.1pp from the expected overlap for the samples of two experiments run in two different programs of exclusive experiments as a function of the number of buckets. Two sizes of experiments are considered 5 and 10 pp of the population.}\label{tab:hyp2}
\end{table}
As expected, the variation in the distribution around the expected value decreases as a function of the number of buckets in the population. With 10000 buckets, the probability of the overlap within 0.1pp from the expected overlap is around 70\% for both experiment sizes. With 100k buckets, the probability has converged to one. In other words, with 100k buckets or larger, the overlap between the samples of experiments run in different programs is very close to the expected value of the overlap, and thus the bucket structure is not severely affecting the properties of the overlap in comparison with sampling of units.

\section{Properties of the difference-in-means estimator of the Population Average Treatment Effect under bucket reuse} \label{sec:everchange}

Let $N$ be the population size, and let $W_i$ be and indicator variable taking the value 0 if unit $i$ is in control and 1 if unit $i$ is in treatment. 
Let $Y^t_i$ denote the outcome of unit $i$ at time point $t$. Use $Y^t_i(W=0)$ and $Y^t_i(W=1)$ to denote the potential outcomes \citep{Rubin2005} of unit $i$ under control and treatment, respectively, for $i=1,...,N$.
The population Average Treatment Effect (ATE) at time point $t$ is defined as
\begin{equation}
ATE_t=\frac{1}{N}\left( \sum_{i=1}^N Y^t_i(W_i=1) - \sum_{i=1}^N Y^t_i(W_i=0)  \right).
\end{equation} 
We include the time point here to emphasize that the effect of one treatment might be different from one time point to another. More specifically, we are not mainly thinking about changes in the ATE due to season or random variation over time, but rather that the population evolves over time and might therefore react differently to a treatment over time. This is especially true in the tech industry where the development is very fast. For example, the effect of improving the resolution of a camera on a smartphone by a factor of 10, probably had a huge impact on customer satisfaction back in 2008. Making the same improvement today, when almost everyone has a super camera in their smartphone, the perceived improvement is probably a lot smaller. Similar changes in the ATE might occur on a shorter time horizon. For example, the increased overall user satisfaction from improving a feature in a mobile app might depend on how easy that feature is to find within the app. If an experiment evaluating an improvement of a feature is run before or after an update that changes the visibility of the same feature, might substantially affect the treatment effect from the improvement. It can be argued that this is actually a different treatment effect, or even an interaction effect. However, from the perspective of the experimenter, it might not be possible to keep track of all possibly relevant changes, and therefore more helpful to simply view the ATE as a process over time . 

Incorporating the possible change in ATE over time is crucial for studying the implications of program of exclusive experiments on the unbiasedness of the difference-in-means ATE-estimator. Let $\mathbf{S}$ be the set of units in a sample. Define the difference in means estimator with two equally sized treatment groups as
\begin{equation}
\widehat{ATE}_t = \frac{2}{N_S}\left(\sum_{i\in \mathbf{S}, W_i=1}Y^t_i(1)-\sum_{i\in \mathbf{S}, W_i=0}Y^t_i(0) \right),
\end{equation}
where $N_S$ is the sample size.
The properties of this estimator under bucket resuse with and without restricted random sampling is the focus of the remainder of this paper.

\subsection{Complete random sampling of Buckets}\label{sec:BR_properties_CR}
Experiments that are run outside of programs of exclusive experiments, i.e., non-exclusive experiments, are run on random samples of buckets from the full population. 
Random sampling of buckets, instead of units, is equivalent to cluster sampling, see, e.g., \cite{lohr2019sampling} for a recent introduction. That is, a bucket is a cluster for all statistical intentions and purposes. There is a related literature on experiments 'embedded' in complex sampling designs that addresses the analysis of randomized experiments based on cluster-sampled samples \citep{horvitz1952, kish1974, vandenBrakel1998, brakel2005} which is directly applicable for bucket reuse under random sampling of buckets. We will go through the key results from this litterature and provide an alternative proof of the unbiasedness of the difference-in-means estimator of the ATE, to build intuition. These results will also serve as a foundation for investigating the properties under restricted random sampling of buckets as imposed by programs of exclusive experiments.

Let the population be split into $B$ equally sized clusters (buckets) of size $N_B$ such that $\sum_{k=1}^K N_B = N_B*B = N.$ Let $\mathbf{S}$ denote a sample (set) of $N_S$ units. Under random sampling of units, there are $\tbinom{N}{N_S}$ ways to draw a sample. Let $\mathcal{S}$ denote the set of all possible samples $\mathbf{S}$ under random sampling of units and denote the cardinality of $\mathcal{S}$ by card($\mathcal{S}$). For all sample sizes $N_S$ such that $N_B \mod N_S=0$, let $\mathcal{S}_B$ denote the set of all possible samples of size $N_S$ under random sampling of buckets, $\mathcal{S}_B$, where the cardinality of this set is given by card($\mathcal{S}_B$)$=\tbinom{B}{\frac{N_SB}{N}}$.
 Let $\mathbf{W}=(W_1,...,W_i,...,W_{N_S})$ be a treatment indicator vector, where element $W_i$ takes the value 0 if unit $i$ is in control and 1 if unit $i$ is in treatment. For simplicity, assume that the sample, once drawn, is split into two equal sized groups ('treatment' and 'control') of size $N_s/2$. Let $\mathcal{W}$ be the set of all possible treatment vectors under random treatment assignments such that $\text{card}(\mathcal{W})=\tbinom{N_S}{N_S/2}$.

In \cite{horvitz1952}, the authors show that the Horvitz-Thompson estimator of the difference-in-means is an unbiased estimator of the population ATE under random sampling of clusters (and more general sampling designs). The Horvitz-Thompson estimator of the population means under treatment or control are given by
\begin{align*}
\hat{\bar{Y}}_{\pi_w}^w=\frac{N_S}{N \frac{N_S}{2}}\sum_{i\in S_s} \frac{Y_i(W=w)}{\pi_i},
\end{align*}
 where $\pi_i$ is the probability of including unit i in the sample. 
Further more, for inference regarding the population ATE, the authors proposed the Horvitz-Thompson version of the t-statistic given by
\begin{equation}\label{eq:hor-t}
\tilde{t}=\frac{\hat{\bar{Y}}_{\pi_1}^1-\hat{\bar{Y}}_{\pi_0}^0}{\sqrt{\text{var}(\hat{\bar{Y}}_{\pi_1}^1-\hat{\bar{Y}}_{\pi_0}^0})}.
\end{equation}
It can be shown that with equally sizes buckets, the ordinary difference-in-means estimator is a special case of the Horvitz-Thompson estimator of the population ATE \citep{horvitz1952}. This has several implication for the properties of the estimator. 
\begin{theorem}\label{thm:unbiased}
Under random sampling of equally sized buckets, with random treatment assignment into two equally sized groups, the sample difference-in-means estimator is an unbiased estimator of the ATE. I.e, 
\begin{equation*}
E[\widehat{ATE}_t] = ATE_t,
\end{equation*}
where expectation is taken over the design space of random samples and treatment allocations.
\end{theorem} 
\begin{proof}
See Appendix \ref{app:unrestricted}.
\end{proof}
Theorem \ref{thm:unbiased} is important as it implies that the difference-in-means estimator in all experiments run outside of programs of exclusive experiments are not biased by the bucket structure. 

Even though the Horovitz-Thompson estimator for the difference-in-mean reduces to the ordinary difference-in-means under equally sized buckets, the variance estimator used for inference needs special attention. The variance of the difference-in-means estimator becomes more intricate under bucket-sampling. This follows from the fact that for any given sample, the covariance between the treatment group means, $\text{cov}(\bar{Y}_s^1  ,\bar{Y}_s^0)$ might not be zero, due to the cluster structure --  which would imply that
$\text{var}(\bar{Y}^1 - \bar{Y}^0)=\text{var}(\bar{Y}^1 ) + \text{var}(\bar{Y}^0) +2\text{cov}(\bar{Y}^1  ,\bar{Y}^0)\neq \text{var}(\bar{Y}^1 ) + \text{var}(\bar{Y}^0).$
\cite{brakel2005} derives the variance of $\text{var}(\bar{Y}_s^1 - \bar{Y}_s^0)$ under complex sampling designs and make the remarkable finding that the variance is only a function of the so-called first order inclusion probabilities \citep{horvitz1952}. That is, even though cluster sampling has a well known effect on the efficiency of the sample mean estimator of the population mean (see, e.g., \cite{sukhatme1984sampling}, \cite{KumarPradhan2007}), the efficiency of the contrasts implied by difference-in-means sample estimators are in most cases not. This implies that under cluster sampling with equal sized cluster, the design-based Horvitz-Thompson variance estimator reduces approximately to the standard variance of the Welch's t-test. This further implies that under equally sized clusters (also called self-weighted sampling designs), the $\tilde{t}$-statistic given in Equation \ref{eq:hor-t} reduces approximately to the t-statistic \citep{vandenBrakel1998}. This means that the effect of the cluster sampling on the sampling distribution of the difference-in-means statistic vanishes when the inclusion probability is equal for all units, and treatment assignment is random in the given sample. This finding is in line with earlier observations by \cite{kish1974}.

It is easy to validate the claims of \cite{brakel2005} using Monte Carlo simulations. All simulations can be replicated using the Julia code in the supplementary files. The following is a small simulation, to build intuition. Outcome data are generated under the null hypothesis of $ATE=0$ according to 
\begin{align}\label{eq:sim:dgpsamp}
Y_i(W=0) &= b*Z_i + X_i
\end{align} 
where $Z, X \sim N_2(\mathbf{1}_2,\mathbf{I}_2)$. The parameters of the simulations are given in Table \ref{tab:sim_sampdist}. For each replication, a population of size 10000 observations is generated according to Equation \ref{eq:sim:dgpsamp}. From each population, 100 independent samples of size 1000 are drawn. In each sample the difference-in-means estimate and the t-statistic is calculated for 100 random treatment assignments.   
\begin{table}[h!]\centering
\begin{tabular}{rc}
Parameter & Values\\
\hline
Number of replications & 100\\
Population size ($N$) & 10000\\
Sample size ($N_S$) & 1000\\
Number of buckets ($B$) & 20\\
Bucket size ($N_B$) & 500\\
Number of random samples from each Population & 100\\
Number of random treatment assignments for each sample & 100\\
Sampling strategies & Random Units, Random clusters/buckets\\
\hline
\end{tabular}
\caption{Parameters of the Monte Carlo Simulation to illustrate the sampling distributions of the difference-in-means estimator and the t-statistic under random sampling of buckets and units, respectively.  }\label{tab:sim_sampdist}
\end{table}
Figure \ref{fig:sim_samdist} displays the results pooled across all populations, samples, and treatment assignments in terms of density plots. Clearly, the t-statistic has the same distribution under random sampling of buckets as under random sampling of units (the lines are exactly on top of each other), even though the sampling distributions of the difference-in-means estimator differs between the sampling strategies. This is exactly in line with \cite{brakel2005} and \cite{kish1974}. Note that, in this example, the sampling variation under random sampling of clusters is smaller than under random sampling of units. This is not a general result, but is a function of the intra-class correlation imposed but the data generating process. Since the intraclass correlation is generally not known and technically difficult to estimate with large populations, we want to have a inference strategy that applies regardless. 
\begin{figure}[h!]\centering
\caption{Density distribution plots for Monte Carlo simulations to illustrate the sampling distributions of the difference-in-means estimator and the t-statistic under random sampling of buckets and units, respectively.}\label{fig:sim_samdist}
\includegraphics[scale=1.2, trim={5.2cm 16cm 4.5cm 4.5cm},page=5, clip=TRUE]{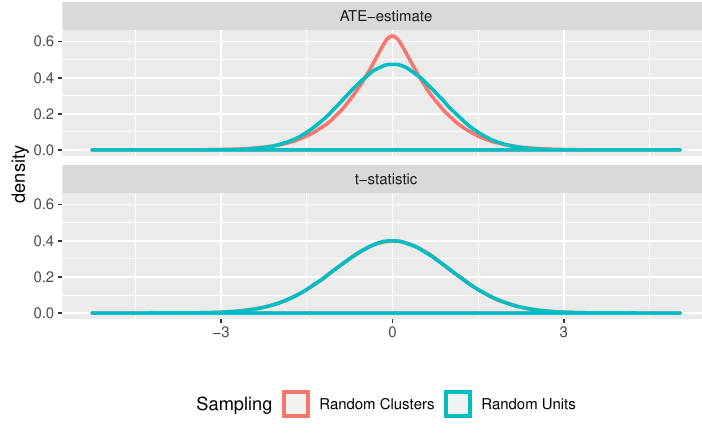}
\end{figure}
In summary, under random sampling of equal sized buckets with random treatment assignment, the inference to the population ATE is equivalent to the inference random sampling of units. This is of critical importance, as it implies that the inference for all experiments that are run outside of exclusivity programs are unaffected by the sampling of buckets. The remainder of this paper will focus on experiments run in programs of exclusive experiments.

\subsection{Restricted random sampling of buckets}
When we run programs of exclusive (non-overlapping) experiments we do not draw random samples from the population of buckets, but from the pool of buckets available at time point $t$. That is, let $\mathcal{B}$ be the set of buckets in the population with $\text{card}(\mathcal{B})=B$. Denote the set buckets available at time point t by $\mathcal{\tilde{B}}_t$, where $\mathcal{\tilde{B}}_t\subseteq \mathcal{B}$.
The problem with randomly sampling buckets from $\mathcal{\tilde{B}}_t$ rather than $\mathcal{B}$ is that the probability of including the units in the bucket in the complement of $\mathcal{\tilde{B}}_t$ is zero. This implies that we can no longer rely on the results of \cite{horvitz1952} and \cite{brakel2005}, since the unbiasedness of the Horvitz-Thompson estimators and the related variance estimators relies on non-zero inclusion probabilities for all units.

Before looking closer at the dependency structures caused by exclusive experiments that in turn imposes restrictions on the random sampling -- we start by noting that this is not an issue related to bucket reuse per se. Although the probability of certain overlaps between consecutive experiments is affected by the bucket structure for small number of buckets (see Section \ref{sec:number_of_buckets}), the dependencies are mainly a consequence of the exclusivity of the experiments. If experiments are run exclusively to each other, the sampling is restricted regardless of the sampling strategy itself. If random sampling of units is used, the random sampling is restricted from $\tbinom{N}{N_S}$ possible samples to $\tbinom{N_{\mathcal{\tilde{B}}}}{N_S}$, where $N_{\mathcal{\tilde{B}}}$ is the number of units in $\mathcal{\tilde{B}}$. Under random sampling of buckets, the corresponding reduction is from $\tbinom{B}{\frac{N_S}{N_B}}$ to $\tbinom{\text{card}(\mathcal{\tilde{B}})}{\frac{N_S}{N_B}}$ possible samples. In a program of exclusive experiments, this restriction is not random with relation the potential outcomes, and therefore imposes bias regardless of sampling unit.


A natural question at this point is: \textit{How} random the sample is under restricted random sampling of buckets? This a sensible question, since the sampling -- although restricted -- has elements of randomness. 
Under unrestricted random sampling of buckets there is no dependency between any samples. In a program of exclusive experiments there is dependency between the samples. In this section we investigate over how long time that dependency reaches, and show that the dependency is a function of how the program is run in terms of amount, sizes, and lengths of the experiments.

Let $\mathbf{B}_t$ be a $B\times 1$ vector where the element $b$ is given by
\begin{align*}
B^t_b&=\begin{cases}
1 & \text{if bucket }b \in \tilde{\mathcal{B}}_t\\
0 & \text{if bucket }b \not\in \tilde{\mathcal{B}}_t\\
\end{cases}.
\end{align*}
To avoid continuous time issues, assume that a the state of a bucket is fixed at a given time point. This assumption is not critical since the number of discrete time periods can be increased arbitrarily to make this true for any setting. 
If the subset $\mathcal{\tilde{B}}_t$ is randomly drawn from $\mathcal{B}$, it holds that 
$\mathbf{B}_t\indep \mathbf{B}_{t-1},\mathbf{B}_{t-2},... ,$ where $\indep$ denotes independence in the sense of \cite{Dawid1979}.
However, in a program of exclusive experiments it is generally the case that
$\mathbf{B}\not\indep \mathbf{B}_{t-1},\mathbf{B}_{t-2},...,$. This is simply a formal way of stating the obvious – the buckets available at time $t$ is not independent of the history of the program of exclusive experiments, and thereby not completely random in the sense that each bucket has the same probability to be sampled at time $t$ given the history of the program up until that time point. 
Without loss of generality, $\tilde{\mathcal{B}}_t$ at time point $t$ can be written as 
\begin{align*}
\tilde{\mathcal{B}}_t&=\Psi(\tilde{\mathcal{B}}_{t-1}; \epsilon, \phi, \theta),
\end{align*}
where, $\Psi$ is some function, $\theta$ is the sampling mechanism of available buckets into a sample, $\epsilon$ governs the probability that a bucket is in an experiment that ends at any given time point, and $\phi$ governs the probability that an experiments starts at any given time point.  
Both $\epsilon$ and $\phi$ are functions of the distribution of lengths, sizes of experiment and amount of experimentation, i.e.
$\epsilon =f(\mathbf{L},\mathbf{ N}, \mathbf{T})$ and $\phi =g(\mathbf{L},\mathbf{ N}, \mathbf{T})$, where $g$ and $h$ are some possibly stochastic functions, $\mathbf{L}$ is the distribution of lengths of experiments in the program, $\mathbf{N}$ is the distribution of sample sizes of experiments in the program, and $\mathbf{T}$ is the average proportion of the population experimented on at any given time point in the program. The set $\mathcal{\tilde{B}}_t$ is in most cases, even under restricted sampling, a partly stochastic function of the state in the previous time period due to the randomness imposed by $\theta$, $\phi$ and $\epsilon$. The sampling mechanism $\theta$ imposes stochastic changes to $\mathcal{\tilde{B}}_t$ over time  by construction. However, we note that in practice also $\phi$ and $\epsilon$ are 'stochastic' in the sense that $\tilde{\mathcal{B}}_t)$ is not strictly determined based on the current state of the exclusive experimentation program. In a large organization like Spotify, even within one experimentation program there are usually several independent teams that experiment with somewhat different designs, and the timing in terms of start and stop is not coordinated. In summary, the state of the program at time point $t$ (in terms of $\tilde{\mathcal{B}}_t$ is not a deterministic function of the history of the program up until that time point -- both by design and by organization.  

We can express the set of currently available buckets as a composite function of the previous sets, i.e.,
\begin{align*}
\tilde{\mathcal{B}}_t&=\Psi(\tilde{\mathcal{B}}_{t-1}; \epsilon, \phi, \theta)\\
\tilde{\mathcal{B}}_t&=\Psi(; \theta, \epsilon, \phi) \circ \Psi(\tilde{\mathcal{B}}_{t-2}; \theta,\epsilon, \phi) \\
\tilde{\mathcal{B}}_t&=\Psi(; \theta,\epsilon, \phi) \circ\Psi(; \theta,\epsilon, \phi) \circ \Psi( \tilde{\mathcal{B}}_{t-3}; \theta,\epsilon, \phi) \\
&\vdots \\
\tilde{\mathcal{B}}_t&=\Psi(; \theta,\epsilon, \phi) \circ\dots\circ \Psi(\tilde{\mathcal{B}}_{1}; \theta, \epsilon, \phi),
\end{align*}
where by construction $\mathcal{\tilde{B}}_1=\mathcal{B}$. In each time step, there is randomness injected by the random sampling.
It is easy to construe programs of exclusive experiments run in such a way that
\begin{align}\label{eq:indep_hist}
\mathcal{\tilde{B}}_t \indep \mathcal{\tilde{B}}_{t-\delta}, \mathcal{\tilde{B}}_{t-\delta -1},...,\mathcal{\tilde{B}}_{1},
\end{align}
for some integer $\delta>0$.
For example, imagine that the amount of the population used for experimentation is small on average and the length of experiments short. In this case, most buckets are available for sampling most of the time and it wouldn't take much time between two time points, $t-\delta$ and $t$, until the dependency between what buckets are available at these time points diminished. That is, the probability of a bucket $b$ being available for sampling at time point $t$, is independent of if the bucket was available or not at time point $t-\delta$. Even in situations where the the amount of experiments is large, the experiments are long, and the sizes of the experiments are large – the stochastic elements that are part of the sampling process will over time decrease the dependency between the sampling at time $t-\delta$ and the sampling at time point $t$, for a large enough $\delta$. \\

\begin{condition}\label{cond_indep}
Let $\delta_{\indep}>0$ be an integer indicating number of time periods
\begin{equation*}
\delta_{\indep}: \tilde{\mathcal{B}}_t \indep \tilde{\mathcal{B}}_{t-\delta_{\indep}}, \tilde{\mathcal{B}}_{t-\delta_{\indep}-1},...,\tilde{\mathcal{B}}_{1} \,\, \forall \,\, t=\delta_{\indep}+1,\delta_{\indep}+2,....
\end{equation*}

\end{condition}
For a $\delta$ fulfilling Condition \ref{cond_indep}, i.e., $\delta_{\indep}$, our sample of buckets at time $t$ is random with respect to the history of the program of exclusive experiments up until time point $t-\delta_{\indep}$. In other words, the sample is not random with respect to the last $\delta_{\indep}$ time points before we drew our sample, but random with respect to everything that happened before that.

To connect Condition \ref{cond_indep} to the bias imposed by the restricted sampling, define the slightly odd difference-in-means estimator 
\begin{equation}
\widehat{ATE^{\tilde{\mathcal{B}_t}}_{t-\delta_{\indep}}}=\frac{1}{N_S/2}\left(\sum_{i: W_i=1, i\in\mathbf{S}, \mathbf{S}\in\mathcal{S}_{\tilde{\mathcal{B}}_t }} Y^{t-\delta_{\indep}}_i(1) - \sum_{i: W_i=0, i\in\mathbf{S}, \mathbf{S}\in\mathcal{S}_{\tilde{\mathcal{B}}_t }} Y^{t-\delta_{\indep}}_i(0)    \right).
\end{equation} 
This is the estimate within a random sample from the available subset of buckets at time $t$, using the potential outcomes from time period $t-\delta_{\indep}$, under random treatment assignment.
\begin{lemma}\label{lemma:indep_delta}
For a $\delta_{\indep}$ fulfilling Condition \ref{cond_indep}, the difference-in-means estimator $\widehat{ATE^{\tilde{\mathcal{B}}_t}_{t-\delta_{\indep}}}$ is an unbiased estimator of  $ATE_{t-\delta_{\indep}}$, i.e., 
\begin{equation*}
E[\widehat{ATE^{\tilde{\mathcal{B}}_t}_{t-\delta_{\indep}}}] = ATE_{t-\delta_{\indep}}
\end{equation*}
\end{lemma}
\begin{proof}
See Appendix \ref{app:restricted}.
\end{proof}
Lemma \ref{lemma:indep_delta} is not very helpful in itself, as we cannot generally observe $Y^{t-\delta_{\indep}}(W)$ at time point $t$. However, this lemma has interesting implications for the bias imposed by the restrictions on the subset of available buckets.
Denote the average treatment effect in  a subset of buckets $\mathcal{B}'$, where $\mathcal{B}'\subseteq \mathcal{B}$, 
\begin{align*}
ATE_t^{\mathcal{B}_t'} = \frac{1}{\text{card}(\mathcal{B}_t')}\left( \sum_{i \in \mathcal{B}'_t} Y^t_i(1)-\sum_{i \in \mathcal{B}'_t} Y^t_i(0) \right).
\end{align*}
\begin{definition}\label{def:ate_tilde}
Define the difference between the ATE's for a subset of buckets $\mathcal{B}_t'$ at two distinct time points $t$ and $t'$, where $t'<t$, as
\begin{equation*}
\widetilde{ATE}^{\mathcal{B}_t'}_{t':t}  \equiv ATE^{\mathcal{B}_t'}_t - ATE^{\mathcal{B}_t'}_{t-t'},
\end{equation*}
where if $\mathcal{B}_t'=\mathcal{B} $ we drop the  superscript, i.e., $\widetilde{ATE}_{t':t}  = ATE_{t} - ATE_{t-t'}$.
\end{definition}
Definition \ref{def:ate_tilde} is a simple decomposition of the average treatment effect for any subset of buckets at any time point into the ATE at a previous time point and a remainder. 
Combining these results yields the following theorem. 
\begin{theorem}\label{thm:bias}
For a $\delta_{\indep}$ fulfilling condition \ref{cond_indep}, the bias in the difference in means estimator caused by sampling from the restricted set of of buckets $\tilde{\mathcal{B}}_t$ instead of $\mathcal{B}$ is given by
\begin{align*}
ATE- E[\widehat{ATE^{\tilde{\mathcal{B}}_t}_t}] &= \widetilde{ATE}_{t-\delta_{\indep}:t} - \widetilde{ATE}^{ \tilde{\mathcal{B}}_t}_{t-\delta_{\indep}:t}.
\end{align*} 
\end{theorem}
\begin{proof}
See Appendix \ref{app:restricted}.
\end{proof}
Theorem \ref{thm:bias} says that in a program of exclusive experiments where there exists a $\delta$ fulfilling Condition \ref{cond_indep}, the bias imposed by the restricted sampling is limited to the difference between the evolution of the ATE in the population contra the bucket of available buckets during the last $\delta_{\indep}$ time points. This gives important insights for running programs of exclusive experiments in practice. If we can empirically estimate $\delta_{\indep}$ for a program, it implies that we only have to address the possible biases caused during the last $\delta_{\indep}$ time points. All effects from experiment before time point $t-\delta_{\indep}$ are uniformly spread across experiments at time point $t$ in expectation. The only remaining bias comes from the divergence between the potential outcomes in $\tilde{\mathcal{B}}_t$ and $\mathcal{B}$ at time $t$. This also implies that if at any time point $t$ it holds that $\tilde{\mathcal{B}}_t=\mathcal{B}$ or the ATE is fixed over time, the difference-in-means estimator is an unbiased estimator of $ATE_t$, as expected. In later sections we present a way to estimate $\delta_{\indep}$ for a given program. One trivial lower bound on $\delta_{\indep}$ can be formulated as the following theorem. 
\begin{theorem}
In a program of exclusive experiments, a $\delta_{\indep}$ fulfilling Condition \ref{cond_indep} is always larger than the length of the longest experiment ran in the program.
\end{theorem}
\begin{proof}
If an experiment of length $D$ starts at time $t$, the set of available buckets will deterministically depend on this experiment until it stops which implies that a $\delta_{\indep}$ fulfilling Condition \ref{cond_indep} can trivially not be smaller than D, which in turns proves the theorem. $\blacksquare$
\end{proof}

\subsubsection{Estimating the dependency length $\delta_{\indep}$}
It is certainly possibly to find mathematical conditions on $\phi$, $\theta$, and $\epsilon$ (including $\mathbf{T}$, $\mathbf{N}$, $\mathbf{L}$) under which there exists finite $\delta$ that fulfills Condition \ref{cond_indep}. However, from an applied perspective we are more interested in finding $\delta_{\indep}$ for our existing exclusive experimentation programs, so that we can manually keep track of the effects in the experiments the last $\delta_{\indep}$ time periods when starting a new experiment to avoid large biases.
Importantly, since $\mathbf{B}$ is a vector of Bernoulli random variables, it holds that $\text{cor}(\mathbf{B}_t, \mathbf{B}_{t-\delta}) \Leftrightarrow \mathbf{B}_t \indep \mathbf{B}_{t-\delta} \Leftrightarrow \tilde{\mathcal{B}}_t \indep \tilde{\mathcal{B}}_{t-\delta} $ \citep{Dai2013}, which implies that it is sufficient to study the correlation to establish independence in practice. In other words, we can learn about $\delta_{\indep}$ in any program of exclusive experiment by studying the correlation between $\mathbf{B}_t$ and $\mathbf{B}_{t-{\delta}}$ for various values of $\delta$. We propose using the estimator
\begin{equation}
\widehat{\delta_{\indep}} = \min(\delta) :  \frac{1}{T^*-\delta+1} \sum_{t=1}^{T-\delta+1} cor^*(\mathbf{B}_{t},\mathbf{B}_{t+\delta})=0,
\end{equation} 
where
\begin{align*}
cor^*(\mathbf{B}_{t},\mathbf{B}_{t+\delta})=\begin{cases}
cor(\mathbf{B}_{t},\mathbf{B}_{t+\delta}) & \text{if } var(\mathbf{B}_{t})\neq 0 \cap var(\mathbf{B}_{t+\delta})\neq 0\\
0 & \text{if }  \mathbf{B}_{t}=\mathbf{}  \text{ or } \mathbf{B}_{t+\delta}=\mathbf{1}\\
1 & \text{if }   \mathbf{B}_{t}= \mathbf{B}_{t+\delta}=\mathbf{0}\\
NA & \text{otherwise}
\end{cases},
\end{align*}
and $T*$ is the number of periods where $cor^*(\mathbf{B}_{t},\mathbf{B}_{t+\delta}) \neq NA$, and $cor(\mathbf{B}_{t},\mathbf{B}_{t+\delta})$ is the sample Pearson correlation coefficient.
The alternative 'correlation' estimator handles the special cases that makes the sample correlation estimator ill defined, but that are logically reasonable in this setting. If any of the vectors are all 1's the correlation is zero since all buckets where available at that time. If all are zero, no buckets are available and the correlation is exactly 1. The estimator $cor^*(\mathbf{B}_{t},\mathbf{B}_{t+\delta})$ is used in the Monte Carlo simulations in the following section.

\section{Practical aspects for running programs of exclusive experiments}\label{sec:practical}

In this section focus on practical aspects of running programs of exclusive experiments. First we discuss the utility of stopping and reshuffling units into new buckets. Second,  we present Monte Carlo simulations of the dependency structure of a experimentation program at Spotify.   

\subsection{Stopping or no stopping}
A question that naturally rises when considering bucket reuse is if there is any need, or at least utility, of stopping all experiments and reshuffle users into new buckets with some cadence. This is an appealing idea: if the program is stopped and the units reshuffled into new buckets using a random salt, the dependency between the state before and after the reshuffle is completely removed. However, the results in this paper indicates that the utility depends on if the bucket reuse is used for non-exclusive experiments or programs of exclusive experiments. 

Given the results in Section \ref{sec:BR_properties_CR}, it is clear that if the random sampling of buckets is not restricted – the bucket structure has no effect on the validity of the treatment effect estimators in the experiments and stopping is never necessary. These results should not be conflated with the effect of the bucket structure on estimators of other estimands such as the population mean. As pointed out by \cite{sukhatme1984sampling, KumarPradhan2007} and others, the efficiency of many estimators under cluster sampling is a function of the intra-class correlation structure imposed by the clusters (buckets). The intra-class correlation is expected to increase after the experimentation program is started since units is a bucket will have more in common on average than users in two different buckets. However, units in a bucket do not share all experiences. Even though units within the same buckets are bounded to go through the same sequence of experiments by construct, they will not go through the same treatments as the treatment assignment is completely random within samples. This is an important difference between buckets and the traditional 'clusters' in the statistical literature. In the statistical literature, a cluster is often a geographically or contextually defined group, that always, regardless of any experimentation, share some environment which makes them more similar than units between two clusters. In the experimentation platform at Spotify, the cluster are formed randomly in the initialization of the experimentation program which implies that the only thing the units in the bucket has in common is the experiences they go through as a consequence of the experimentation program. As a consequence, all units in the control group of any experiments at a given time point are sharing the same experience, which in turn implies that approximately half of the users in the population share the same experience all the time \footnote{Assuming even splits into treatment and control which is often selected for efficiency reasons.}. This should effectively decrease the heterogeneity among the buckets as opposed to traditional clusters. 

Under restricted random sampling of buckets, as imposed by a program of exclusive experiments, if we 'need' to stop or not depends on what $\delta$ that fulfils Condition \ref{cond_indep}, and what is considered a reasonable length of dependency between experiments in terms of time.  Stopping and reshuffling can make Condition \ref{cond_indep} true for $\delta=1$ at any time point, but the benefit only applies to the first round of experiments started after the stop. Noteworthy, it is sufficient to stop all experiments in the program to achieve this effect – reshuffling is not needed. Once all experiments are stopped,  there is no dependency between the availability for sampling at that time point and any time point before that. To substantially benefit from stopping it needs to be done more often than the corresponding $\delta_{\indep}$ for the given program. We see three plausible ways of handling the dependency between experiments in programs of exclusive experiments in practice, listed below.
\begin{enumerate}
\item Estimate $\delta_{\indep}$ and monitor/adjust for dependencies and possible carry-over bias within that time frame.
\item Tune the amount, sizes, and lengths of experiments ($\mathbf{T}$, $\mathbf{N}$,  and $\mathbf{L}$) to decrease $\delta_{\indep}$ to an acceptably small value.
\item Systematically stop all experimentation with the cadence of the desired $\delta$ to enforce an upper bound on the dependency in terms of time.
\end{enumerate}
At Spotify we have opted for the first alternative since development flexibility is crucial to our organization. Given that our $\delta_{\indep}$ seems to be quite small (see simulations in the following section), it is infeasible to stop more often than that. We often run experiments over several months, and do not want to restrict the starting and stopping of such experiments to be bounded by this.

Another argument for stopping and reshuffling, that has come up in previous discussions, is that the bucket structure could increase the carry-over bias in itself, by increasing the heterogeneity in the population over time. 
Recall the Section \ref{sec:coordexp} and the concepts of paths. By sampling buckets instead of users, we further limit the maximum possible number of paths in our program from one path per unit ($N$) to one path per bucket ($B$). Indeed, there are also several users within each path implying that each path will make up a larger part of the sample for a certain experiment than would be the case with N paths. We note that the $B$ paths are 'randomly' selected in the sense that there is no systematic way in which the bucket structure imposes selecting paths that are associated with deviant treatment effects. Moreover, for programs where $\delta_{\indep}$ is reasonably small, the paths that occurs are in fact random over time, which makes sequence of unfortunate events as likely as they are empirically -- no more no less. That is, the consequence of sampling buckets rather than units is that we randomly draw $B$ paths, instead of randomly draw N paths from the incredibly large set of possible paths\footnote{Assuming a program with many experiments over time.}. The impact in future experiments of randomly hitting a path associated with an exceptionally strong interaction effects from the sequence of experiments is larger since more users are in the path and they will appear in the same sample. However, since the treatment assignment will randomly spread these users into the treatment groups, the increased impact should be small in most conceivable cases.

In summary, assuming that the number of buckets is sufficiently large in relation to the smallest experiments (Section \ref{sec:number_of_buckets}) and if the goal is to estimate average treatment effects: We see no reason from a statistical point of view to reshuffle users into new buckets. Stopping all experimentation always resets the dependency between experiments which decreases the bias in experiments started close in time to the restart, which implies that stopping is a way of tuning the dependency.

\subsection{Simulations}\label{sec:sim_corr}
\FloatBarrier

 We can learn about $\delta_{\indep}$ by running Monte Carlo simulations. All simulations can be replicated using the Julia code in the supplementary files. Six settings are considered, displayed in Table \ref{tab:sim_corr}. For each setting 100 random starting points are considered. For each starting point, the average number of experiments needed to obtain $\mathbf{T}$ amount of experiments are started. For these experiments a random sample of sizes are drawn, and the corresponding amount of buckets are then sampled to each experiment. The lengths for these experiments a random sample of lengths are drawn, and it is noted for how long the buckets will be occupied. Based on each starting point 10000 simulations of 90 days are performed. That is, at each time point the simulation checks the proportion of buckets that are assigned to an experiment. If the proportion is smaller than $\mathbf{T}$, it starts as many experiments that are needed to achieve $\mathbf{T}$ experimentation on average with random lengths and sizes. 
 
 The ATE per bucket at time point t is generated by drawing a random sample of $B$ effects from a Normal distribution with mean 3 and variance 2.  The distributions of lengths considered in the simulations are given by $L_1=$ (1, 2, 3, 7, 7, 8, 8, 8, 12, 13, 14, 14, 14, 14, 15, 21, 21, 21, 30 ) and $L_2=(21, 28)$. The distributions of sizes of experiments considered in the simulation are given by $N_1=(2\%, 2\%,2\%, 2\%, 5\%, 5\%, 8\%, 9\% , 10\%, 10\% )$ and $N_2=(20\%, 25\% )$. The distributions of lengths and sizes, $L_1$ and $N_1$, are approximately those of the empirical experimentation in the Home-screen experiment program at Spotify during the spring of 2020. The size and length of each started experiments is drawn uniformly from $\mathbf{N}$ and $\mathbf{L}$, respectively. Each day it is noted which buckets are available for sampling, which buckets are sampled, and the $\widehat{ATE^{\tilde{\mathcal{B}}_t}_{t-\delta_{\indep}}}$ estimate for that sample is recorded. That is, in addition to $\mathbf{B}_t$ we also collect the $B \times 1$ vector $\mathbf{A}_t$ where the elements
\begin{align*}
A^t_b&=\begin{cases}
1 & \text{if bucket }b \text{ is sampled at time }t\\
0 & \text{if bucket }b \text{ is not sampled at time }t
\end{cases},
\end{align*}
to enable visualization of how the correlation between the samples depend on the availability of buckets.  
\begin{table}[hbt!]\centering
\begin{tabular}{rcccccc}
&\multicolumn{6}{c}{Setting}\\
Parameter & 1 &  2&3&4&5&6\\
\hline
Experiment length distribution ($\mathbf{L}$) & $L1$& $L1$& $L_2$& $L1$& $L_2$& $L_2$\\
Experiment size distribution ($\mathbf{N}$) & $N_1$& $N_1$& $N_1$& $N_2$ & $N_2$ & $N_2$ \\
Avg. prop. of population experimented on($\mathbf{T}$) & 90\% & 50\% & 90\% & 90\% & 90\% & 50\% \\
Number of Buckets & \multicolumn{6}{c}{10000}\\
Number of Days & \multicolumn{6}{c}{90}\\
Number of Starting points &  \multicolumn{6}{c}{50}\\
Number of Replications per starting point &  \multicolumn{6}{c}{10000}\\
\hline
\end{tabular}
\caption{Parameters of the Monte Carlo Simulation. For parameters with lists of values, a value is randomly (uniformly) drawn when a new experiment is started.}\label{tab:sim_corr}
\end{table}
The evaluation metrics for the simulation are given by
\begin{align*}
\text{Available Buckets Correlation} &= cor*(\mathbf{B}_1,\mathbf{B}_\delta)\\
\text{Sampling Correlation} &= cor*(\mathbf{A}_1,\mathbf{A}_\delta)\\
\text{ATE1 Bias} &= ATE_1 - \widehat{ATE^{\mathcal{B}_\delta}_{
1}}.
\end{align*}
The ATE1 Bias measure is included here to confirm Lemma \ref{lemma:indep_delta} which in turn is key to the proof of Theorem \ref{thm:bias}. 

Figure \ref{fig:sim_corr1} and \ref{fig:sim_corr2} display the three outcome metrics across starting points and replications for various values of $\delta$ for Setting 1-6. There are many things to learn from these simulations. Lemma \ref{lemma:indep_delta} is confirmed by all simulations as the ATE1 bias goes to zero as the correlation between the available buckets goes to zero. For setting 1 and 2 (Figure \ref{fig:sim_corr1}), which resembles the empirical experimentation at Spotify, the correlation in sampling and availability is quite small after a short time period\footnote{The peaks are artefacts form that the $\delta$'s are multiplicatives of the lengths of the experiments.}. As expected, availability is positively correlated right after the starting point, whereas the sampling correlation is negative. The buckets that were sampled at time point 1 cannot be sampled again until the initial experiments are stopped, whereas the buckets that were available on time point 1 are available after that unless sampled into a new experiment.  For Setting 1 a plausible $\widehat{\delta_{\indep}}$ seems to be around 6-8 weeks. However, then the amount of experimentation is lowered to 50\% on average, this ATE1 bias is close to zero already after 4 weeks. This is expected, when the amount of experimentation is low, there is always many buckets to randomly sample from which implies a lot of spreading of old experiment samples into new samples. For settings 3 and 4 (Figure \ref{fig:sim_corr2}), the length and size distribution are altered, respectively, to give information about the relative importance of the lengths and the sizes of experiments with respect to $\delta_{\indep}$. In setting three, the lengths of the experiments are all 3 or 4 weeks. This has impact on $\widehat{\delta_{\indep}}$ which seems to be more than 12 weeks. In Setting 4, where the lengths are according to $L_1$ but sizes are 9\% and 10\%, the dependency goes down quickly. That means that going from $L_1$ to $L_2$ seems more problematic than going from $N_1$ to $N_2$, in terms of bucket availability and sampling dependencies.
In settings 5 and 6, the experiments are both long and large. Perhaps counter intuitively, the dependency is not stronger than in settings 4 and 5, and the size of the biases is quite small. However, this is reasonable since each experiment makes up a quite a large part of the population. If each experiment was of size 100\% of the population, the dependency between two experiments would be zero as 'all' experiments has to be stopped until a new can be started. It is clear that under settings 5 and 6, the correlation pattern is very predictable, as might be expected given the hard restrictions on when an experiment can start and when it will end given when it starts. Also with $L_2$ and $N_2$ the dependencies and bias can be effectively decreased by lowering the traffic, as is done in Setting 6.
\begin{figure}[hbt!]\centering
\caption{The ATE1 bias, bucket availability correlation, and bucket sampling correlation plotted as a function of $\delta$, for settings 1-2. The parameters for each setting are given in Table \ref{tab:sim_corr}. }\label{fig:sim_corr1}
\includegraphics[scale=1.2, trim={5cm 16.1cm 3cm 4.5cm},page=6, clip=TRUE]{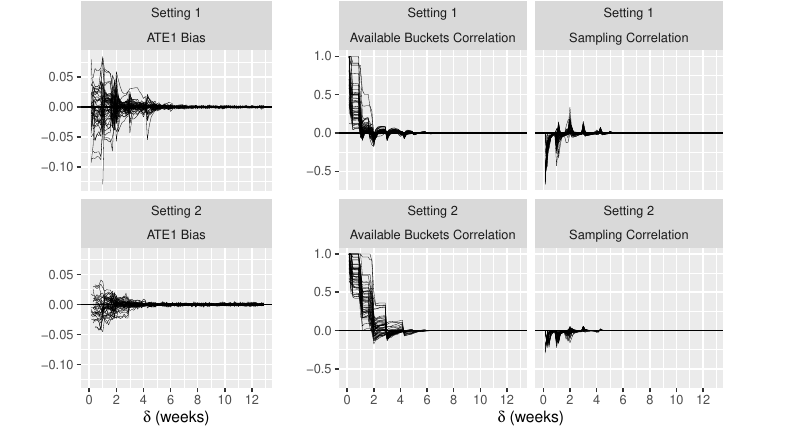}
\end{figure}
\begin{figure}[hbt!]\centering
\caption{The ATE1 bias, bucket availability correlation, and bucket sampling correlation plotted as a function of $\delta$, for settings 3-6. The parameters for each setting are given in Table \ref{tab:sim_corr}. }\label{fig:sim_corr2}
\includegraphics[scale=1.2, trim={5cm 8.5cm 3cm 4.5cm},page=7, clip=TRUE]{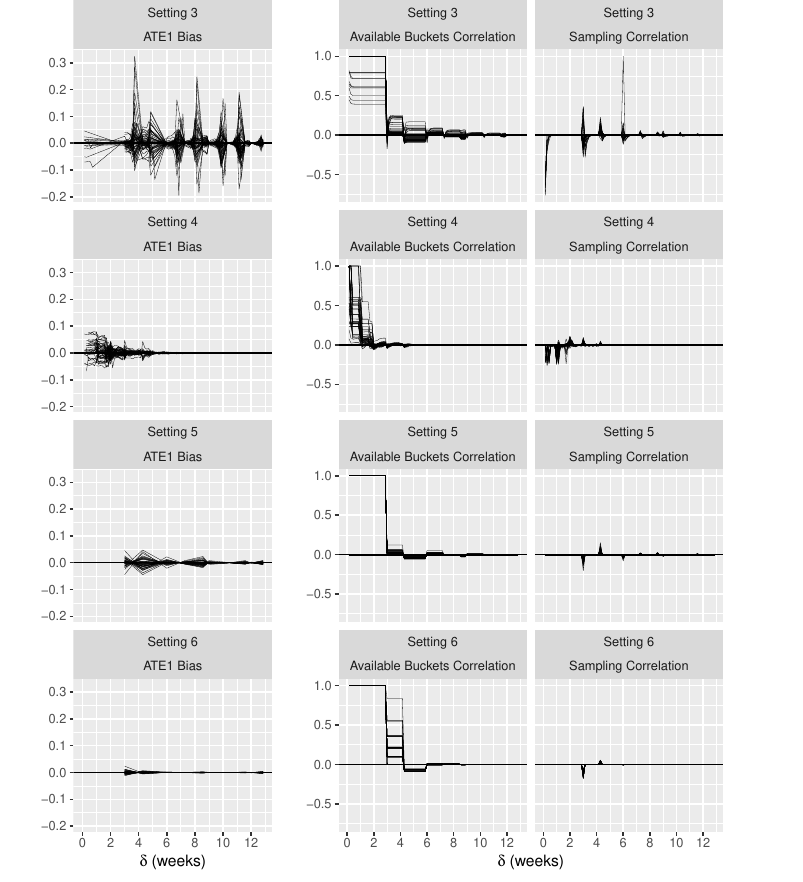}
\end{figure}
Appendix \ref{app:simulations} contains results from the same six settings but with 100 buckets instead of 10000. These results indicate that the bias decreases when the number of buckets increase. This is expected given the data generating process used here. The heterogeneity is fixed, implying that when the number of buckets increase there are more buckets with similar ATE's, which yields less heterogeneity across samples containing many buckets and therefore smaller bias on average.

In summary, the amount of traffic and length of the experiments seem more important than experiment sizes when it comes to decreasing $\delta_{\indep}$. If the program is run with a mix of different sizes below 10\% and of various lengths, experimenting continuously on up to 90\% of the population doesn't cause severe dependencies. In our empirical setting, a reasonable $\delta_{\indep}$ seems to be 8 weeks even with high amount of continuous experimentation.

\FloatBarrier
\section{Discussion and Conclusions}\label{sec:disc}
In this paper we have investigated the properties of Bucket Reuse as a coordinated sampling strategy for experiment programs, both exclusive and non-exclusive. The key results are the following. For non-exclusive experiments sampling buckets is equivalent with traditional cluster sampling. This implies that bucket reuse with a sufficient number of buckets does not affect the inference to the population average treatment effect (ATE) for non-exclusive experiments. Moreover, it is shown that programs of exclusive experiments introduce bias, regardless of if sampling is performed on the unit level or the bucket level. It is shown this bias is a function of the history of the experimentation program, but in most cases only the recent history. That is, the sampling at a given time point $t$ is independent of the sampling at a time point $t-\delta$ for a sufficiently large $\delta$. This implies that the bias is limited to the effects of the experimentation program during the last $\delta$ time points leading up to the sampling. A way of estimating the length of this dependency is proposed, and several ways of handling this issue are discussed. These results should be highly relevant for any coordinated experimentation program regardless of coordination and sampling implementation. The results regarding the bias imposed by restricted sampling applies to sampling on unit and bucket level. These results are for example applicable to experiments run in the system proposed by \cite{Tang2010}.

Based on the results of this paper, the experimentation platform at Spotify will migrate to coordination using Bucket Reuse with 1M (1,000,000) buckets. That is, all users are hashed into 1M buckets, and these buckets are used for all experiments in all experimentation programs, exclusive and non-exclusive. Although we establish results that indicate that smaller numbers of buckets can have sufficiently nice statistical properties, it should be clear that the larger the number of buckets, the better. Even though the inference for the average treatment effect is unaffected by the bucket sampling, it is well known that the effect of cluster sampling on other estimands decreases when the number of buckets increases \citep{KumarPradhan2007}. That is, imposing a bucket structure is not desirable from a statistical perspective, but a technical necessity. The choice of 1M buckets was made because it is close to the largest number of buckets we can have while still keeping the selected buckets within an executable script stored in a database, without having to resort to blob storage. 1M bucket implies that randomization is not severely restricted by the bucket structure for experiments of sizes down to 0.1\% of the population, which is well under our smallest sample sizes in practice. In programs of exclusive experiments, we will also monitor the correlation of availability empirically for different experimentation programs and provide information and recommendation to the experiments accordingly.  

\clearpage
\pagebreak
\bibliography{library}

\appendix
\section{Appendix}
All theorems are repeated here for convenience.
\subsection{Complete random sampling of Buckets}\label{app:unrestricted}

\begin{theorem}
Under random sampling of equally sized buckets, with random treatment assignment into two equally sized groups, the sample difference-in-means estimator is an unbiased estimator of the ATE. I.e, 
\begin{equation*}
E[\widehat{ATE}] = ATE,
\end{equation*}
where expectation is taken over the design space of random samples and treatment allocations.
\end{theorem}

\begin{proof}
Under random sampling of equally sized buckets, it follows that $\pi_i=\frac{N_S}{N}\,\,\forall\,\, i=1,...,N$, which implies that the Horvitz-Thompson estimator simplifies as
 \begin{align}\label{eq:horvits_mean}
\hat{\bar{Y}}_{\pi_w}^w=\frac{N_S}{N \frac{N_S}{2}} \frac{N}{N_S} \sum_{i\in S_s} Y_i(W=w)=\frac{1}{\frac{N_S}{2}}\sum_{i\in S_s} Y_i(W=w)=\bar{Y}^w,
\end{align}
which is simply the sample means of the groups. Unbiasedness follows directly from the results in \cite{horvitz1952}, but we give an alternative proof here to build intuition. 
We drop the superscript $t$ on the ATE to easy notation. Denote the difference-in-means estimator of the average treatment effect 
\begin{equation}
\widehat{ATE}= \frac{1}{N_s/2} \left( \sum_{i:i\in \mathbf{S}\cap W_i=1} Y_i(W_i) -\sum_{i:i\in \mathbf{S}\cap W_i=0} Y_i(W_i)       \right).
\end{equation}
Enumerate all possible sample under random sampling of buckets with $\mathbf{S}_s$ where $s=1,...,\text{card}(\mathcal{S}_B)$.  Moreover, enumerate all possible treatment assignments over random treatment assignments by $W^j$ where $j=1,...,\text{card}(\mathcal{W})$. This implies that we can write one single estimate as
\begin{equation}
\widehat{ATE}_{s,j}= \frac{1}{N_s/2} \left( \sum_{i:i\in \mathbf{S}_s\cap W^j_i=1} Y_i(W_i^j) -\sum_{i:i\in \mathbf{S}_s\cap W^j_i=0} Y_i(W_i^j)       \right).
\end{equation}
The expected value of $\widehat{ATE}_B$, where the subscript $B$ indicates random sampling of buckets from $\mathcal{B}$, is given by
\begin{align*}
E[\widehat{ATE}_B] &= \frac{1}{\text{card}(\mathcal{S}_B)} \sum_{s=1}^{\text{card}(\mathcal{S}_B)}\left(  
\frac{1}{\text{card}(\mathcal{W})} \sum_{j=1}^{\text{card}(\mathcal{W})} \widehat{ATE}_{s,j} 
\right)\\
&= \frac{1}{\text{card}(\mathcal{S}_B)} \frac{1}{\text{card}(\mathcal{W})}  \sum_{s=1}^{\text{card}(\mathcal{S}_B)}
\sum_{j=1}^{\text{card}(\mathcal{W})} \widehat{ATE}_{s,j} \\
&= \frac{1}{\text{card}(\mathcal{S}_B)} \frac{1}{\text{card}(\mathcal{W})} \frac{2}{N_S}  \left(  \sum_{s=1}^{\text{card}(\mathcal{S}_B)}
\sum_{j=1}^{\text{card}(\mathcal{W})} \sum_{i:i\in \mathbf{S}_s\cap W^j_i=1} Y_i(W_i^j) -\sum_{s=1}^{\text{card}(\mathcal{S}_B)}
\sum_{j=1}^{\text{card}(\mathcal{W})} \sum_{i:i\in \mathbf{S}_s\cap W^j_i=0} Y_i(W_i^j)       \right),
\end{align*}
where the last step follows from equally sized buckets. 
Due to the symmetry of the random sampling of buckets and the equal bucket size, each bucket (and thereby unit) will be in equally many samples. Moreover, in each sample, each unit is in the treatment and control groups equally many times, respectively, due to the mirror property of randomization distributions \citep{johansson2020, nordin2020}. This implies that we are simply adding and subtracting the value of each unit several times. It follows that,  
\begin{align*}
E[\widehat{ATE}_B] &=  \frac{1}{\text{card}(\mathcal{S}_B)} \frac{1}{\text{card}(\mathcal{W})} \frac{2}{N_S}  \left(  \sum_{s=1}^{\text{card}(\mathcal{S}_B)}
\sum_{j=1}^{\text{card}(\mathcal{W})} \sum_{i:i\in \mathbf{S}_s\cap W^j_i=1} Y_i(W_i^j) -\sum_{s=1}^{\text{card}(\mathcal{S}_B)}
\sum_{j=1}^{\text{card}(\mathcal{W})} \sum_{i:i\in \mathbf{S}_s\cap W^j_i=0} Y_i(W_i^j)       \right)\\
&= \frac{1}{\text{card}(\mathcal{S}_B)} \frac{1}{\text{card}(\mathcal{W})} \frac{2}{N_S}  \left(  \xi_B \sum_{i=1}^N Y_i(1) -\xi_B \sum_{i=1}^N Y_i(0)       \right).
\end{align*}
The number of time that each unit will be in the control and the treatment group across all possible samples and treatment assignments is given by
\begin{align*}
\xi_B &=  \frac{\text{card}(\mathcal{W})}{2} \times \binom{B-1}{\frac{N_S}{N_B}-1},
\end{align*}
which implies
\begin{align*}
E[\widehat{ATE}_B] &=  \frac{1}{\text{card}(\mathcal{S}_B)} \frac{\binom{K-1}{\frac{N_S}{N_B}-1}}{N_S}  \left(  \sum_{i=1}^N Y_i(1) -\sum_{i=1}^N Y_i(0)       \right).
\end{align*}
We note that
\begin{align*}
\text{card}(\mathcal{S}_B) = \frac{B}{\frac{N_S}{N_B}} \binom{B-1}{\frac{N_S}{N_B}-1}
\end{align*}
which gives the final expression
\begin{align*}
E[\widehat{ATE}_B] &=   \frac{N_S}{N_B B\binom{K-1}{\frac{N_S}{N_B}-1}} \frac{\binom{K-1}{\frac{N_S}{N_B}-1}}{N_S}  \left(  \sum_{i=1}^N Y_i(1) -\sum_{i=1}^N Y_i(0)       \right)\\
&= \frac{1} {N_B B} \left(  \sum_{i=1}^N Y_i(1) -\sum_{i=1}^N Y_i(0)       \right)\\
&= \frac{1} {N} \left(  \sum_{i=1}^N Y_i(1) -\sum_{i=1}^N Y_i(0)       \right)=ATE.
\end{align*}
Which is in line with \cite{horvitz1952} as expected.  $\blacksquare$

\end{proof}

\pagebreak

\subsection{Restricted random sampling of buckets}\label{app:restricted}

\begin{lemma}
For a $\delta_{\indep}$ fulfilling Condition \ref{cond_indep}, the the difference in means estimator $\widehat{ATE^{\tilde{\mathcal{B}}_t}_{t-\delta_{\indep}}}$ is an unbiased estimator of  $ATE_{t-\delta_{\indep}}$, i.e., 
\begin{equation}
E[\widehat{ATE^{\tilde{\mathcal{B}}_t}_{t-\delta_{\indep}}}] = ATE_{t-\delta_{\indep}}
\end{equation}

\end{lemma}

\begin{proof}
Since a program of exclusive experiment is expected to change user behaviour and therefore reactions to future changes, it generally holds that 
\begin{equation}
Y_t(0), Y_t(1) \not\indep \tilde{\mathcal{B}}_t.
\end{equation}
However, it is also the case that the only dependency between the potential outcomes and the set of available buckets is captured by the history of the available buckets. The samples are random from the set of available buckets -- if different subsets have different experiences that affects their behaviour such that the ATE changes, the dependency between the sample and these changes are completely described by the history of $\tilde{\mathcal{B}}_t$. It follows from Condition \ref{cond_indep} that
\begin{equation}
\tilde{\mathcal{B}}_t \indep \tilde{\mathcal{B}}_{t-\delta}, \tilde{\mathcal{B}}_{t-\delta-1},...,\tilde{\mathcal{B}}_{1} \Rightarrow Y_{t-\delta_{\indep}}(0), Y_{t-\delta_{\indep}}(1) \indep \tilde{\mathcal{B}}_t.
\end{equation} 
In other words, in relation to the potential outcomes at time $t-\delta_{\indep}$, the subset $\tilde{\mathcal{B}}_t$ is a random subset from $\mathcal{B}$. One way to understand this is that from a randomization perspective it is equivalent to 1. randomly selecting a subset at time $t$, randomly sample from that subset, and finally randomly assign the treatment, and, 2. semi-randomly selecting subsets in steps until the subset is independent from the starting set (at step $\delta_{\indep}$), and then sample and assign treatment randomly. The semi-random subsetting in $\delta_{\indep}$ steps is essentially an ineffective method for randomly drawing a subset. The important practical difference between 1 and 2 is that in a program of exclusive experiments, the potential outcomes are expected to change as a function of the steps of subsetting in 2. For this reason, performing 1 at time $t$ yields a subset that is independent of the outcome at time $t$ and all other time periods, whereas 2 yields a subset that is independent in relation to the potential outcomes only at time $t-\delta_{\indep}$ and before that. \\\\
Since $\tilde{\mathcal{B}}_t$ is a random subset of buckets in relation to the potential outcomes at time $t-\delta_{\indep}$ under Condition \ref{cond_indep}, it only remains to prove that the difference-in-means estimator is an unbiased estimator of $ATE_{t-\delta_{\indep}}$ under random 'sampling' of a subset, random sampling within the subset and random treatment assignment within the sample to prove Lemma \ref{lemma:indep_delta}.
Enumerate the possible subsets,$\tilde{\mathcal{B}}_t^b$, of a given size as  $b=1,...,\binom{B}{\text{card}(\tilde{\mathcal{B}}_t)}$ ,where each subset is equally probable in relation to the potential outcomes at time $t-\delta_{\indep}$. In each subset, enumerate all possible samples $\mathbf{S}_s$ as$s=1,...,\binom{\text{card}(\tilde{\mathcal{B}}_t)}{\frac{N_S}{N_B}}$.   This implies that, using results from the proof of Theorem \ref{thm:unbiased}
\begin{align*}
E[\widehat{ATE^{\tilde{\mathcal{B}}_t}_t}] &= \frac{1}{\binom{B}{\text{card}(\tilde{\mathcal{B}}_t)}} \sum_{b=1}^{\binom{B}{\text{card}(\tilde{\mathcal{B}}_t)}} \frac{1}{\text{card}(\mathcal{S}_{\tilde{\mathcal{B}}^b_t})} \frac{1}{\text{card}(\mathcal{W})} \frac{2}{N_S} \\& \left(  \sum_{s=1}^{\text{card}(\mathcal{S}_{\tilde{\mathcal{B}}^b_t})}
\sum_{j=1}^{\text{card}(\mathcal{W})} \sum_{i:i\in \mathbf{S}_s\cap W^j_i=1} Y_{i,t-\delta_{\indep}}(W_i^j) -\sum_{s=1}^{\text{card}(\mathcal{S}_B)}
\sum_{j=1}^{\text{card}(\mathcal{W})} \sum_{i:i\in \mathbf{S}_s\cap W^j_i=0} Y_{i, t-\delta_{\indep} }(W_i^j)       \right)\\
&= \frac{1}{\binom{B}{\text{card}(\tilde{\mathcal{B}}_t)}} \frac{1}{\text{card}(\mathcal{S}_{\tilde{\mathcal{B}}_t})} \frac{1}{\text{card}(\mathcal{W})} \frac{2}{N_S}  \\&\left( \sum_{b=1}^{\binom{B}{\text{card}(\tilde{\mathcal{B}}_t)}}  \sum_{s=1}^{\text{card}(\mathcal{S}_{\tilde{\mathcal{B}}^b_t})}
\sum_{j=1}^{\text{card}(\mathcal{W})} \sum_{i:i\in \mathbf{S}_s\cap W^j_i=1} Y_{i,t-\delta_{\indep}}(W_i^j) - \sum_{b=1}^{\binom{B}{\text{card}(\tilde{\mathcal{B}}_t)}} \sum_{s=1}^{\text{card}(\mathcal{S}_B)}
\sum_{j=1}^{\text{card}(\mathcal{W})} \sum_{i:i\in \mathbf{S}_s\cap W^j_i=0} Y_{i, t-\delta_{\indep} }(W_i^j)       \right)
\end{align*}
where, similar to the proof of Theorem \ref{thm:unbiased}, each unit is included in equally many subsets and samples, and within each sample each unit is included in the treatment group and the control group equally many times. It follows that 
\begin{align*}
E[\widehat{ATE^{\tilde{\mathcal{B}}_t}_t}] &= \frac{1}{\binom{B}{\text{card}(\tilde{\mathcal{B}}_t)}} \frac{1}{\text{card}(\mathcal{S}_{\tilde{\mathcal{B}}_t})} \frac{1}{\text{card}(\mathcal{W})} \frac{2}{N_S} \left( \tilde{\xi} \sum_{i=1}^N Y_{i, t-\delta_{\indep} }(1) -\tilde{\xi}\sum_{i=1}^N Y_{i, t-\delta_{\indep} }(0)       \right),
\end{align*}
where
\begin{equation}
\tilde{\xi} = \frac{\text{card}(\mathcal{W})}{2} \times \binom{\text{card}(\tilde{\mathcal{B}}_t)-1}{\frac{N_S}{N_B}-1}\times \binom{B-1}{\text{card}(\tilde{\mathcal{B}}_t)-1}.
\end{equation}
Note that
\begin{align*}
\text{card}(\mathcal{S}_{\tilde{\mathcal{B}}_t}) &= \frac{\text{card}(\tilde{\mathcal{B}}_t)}{\frac{N_S}{N_B}} \binom{\text{card}(\tilde{\mathcal{B}}_t)-1}{\frac{N_S}{N_B}-1}\\
\binom{B}{\text{card}(\tilde{\mathcal{B}}_t)}&= \frac{B}{\text{card}(\tilde{\mathcal{B}}_t)}\binom{B-1}{\text{card}(\tilde{\mathcal{B}}_t)-1}.
\end{align*}
Putting this together it follows that
\begin{align*}
E[\widehat{ATE^{\tilde{\mathcal{B}}_t}_t}] &= \frac{1}{\binom{B}{\text{card}(\tilde{\mathcal{B}}_t)}} \frac{1}{\text{card}(\mathcal{S}_{\tilde{\mathcal{B}}_t})} \frac{1}{\text{card}(\mathcal{W})} \frac{2}{N_S} \left(  \sum_{i=1}^N Y_{i, t-\delta_{\indep} }(1) -\sum_{i=1}^N Y_{i, t-\delta_{\indep} }(0)     \right)\\
&=\frac{1}{\frac{B}{\text{card}(\tilde{\mathcal{B}}_t)}} \frac{1}{\frac{\text{card}(\tilde{\mathcal{B}}_t)}{\frac{N_S}{N_B}} }  \frac{1}{N_S} \left(  \sum_{i=1}^N Y_{i, t-\delta_{\indep} }(1) -\sum_{i=1}^N Y_{i, t-\delta_{\indep} }(0)     \right)\\
&=\frac{1}{N_BB}  \left( \sum_{i=1}^N Y_{i, t-\delta_{\indep} }(1) -\sum_{i=1}^N Y_{i, t-\delta_{\indep} }(0)     \right)\\
&=\frac{1}{N} \left( \sum_{i=1}^N Y_{i, t-\delta_{\indep} }(1) -\sum_{i=1}^N Y_{i, t-\delta_{\indep} }(0)     \right)=ATE_{t-\delta_{\indep}}\,\, \blacksquare
\end{align*}

\end{proof}


\begin{theorem}
For a $\delta_{\indep}$ fulfilling condition \ref{cond_indep}, the bias in the difference in means estimator caused by sampling from the restricted set of the population of buckets given by $\tilde{\mathcal{B}}$ is given by
\begin{align*}
ATE- E[\widehat{ATE^{\tilde{\mathcal{B}}_t}_t}] &= \widetilde{ATE}_{t-\delta_{\indep}:t} - \widetilde{ATE}^{ \tilde{\mathcal{B}}_t}_{t-\delta_{\indep}:t}.
\end{align*} 
\end{theorem}

\begin{proof}
Note that from Definition \ref{def:ate_tilde} it follows that 
\begin{align*}
ATE - E[\widehat{ATE^{\tilde{\mathcal{B}}}_t}] &= ATE_{t-\delta_{\indep}}+\widetilde{ATE}_{t-\delta_{\indep}:t} - ATE_{t-\delta_{\indep}}^{ \tilde{\mathcal{B}}_t} + \widetilde{ATE}_{t-\delta_{\indep}:t}^{ \tilde{\mathcal{B}}_t}.
\end{align*} 
By Lemma \ref{thm:unbiased} the difference-in-means estimator under random sampling of bucket from $\tilde{\mathcal{B}}_t$ is an unbiased estimator of  $ATE_{t-\delta_{\indep}}$ which implies that $ATE_{t-\delta_{\indep}}^{ \tilde{\mathcal{B}}_t} = ATE_{t-\delta_{\indep}} $ which directly gives 
\begin{align*}
ATE - E[\widehat{ATE^{\tilde{\mathcal{B}}}_t}] &= ATE_{t-\delta_{\indep}}+\widetilde{ATE}_{t-\delta_{\indep}:t} - ATE_{t-\delta_{\indep}} + \widetilde{ATE}_{t-\delta_{\indep}:t}^{ \tilde{\mathcal{B}}_t}\\
&= \widetilde{ATE}_{t-\delta_{\indep}:t} - \widetilde{ATE}^{ \tilde{\mathcal{B}}_t}_{t-\delta_{\indep}:t} \,\,\blacksquare
\end{align*} 
\end{proof}

\section{Simulations}\label{app:simulations}
Here the simulation from Section \ref{sec:sim_corr} is repeated but with 100 buckets instead of 10000. All simulations can be replicated using the Julia code in the supplementary files. The parameters are repeated in Table \ref{tab:sim_corr2} for convenience.
\begin{table}[h!]\centering
\begin{tabular}{rcccccc}
&\multicolumn{6}{c}{Setting}\\
Parameter & 1 &  2&3&4&5&6\\
\hline
Experiment length distribution ($\mathbf{L}$) & $L1$& $L1$& $L_2$& $L1$& $L_2$& $L_2$\\
Experiment size distribution ($\mathbf{N}$) & $N_1$& $N_1$& $N_1$& $N_2$ & $N_2$ & $N_2$ \\
Avg. prop. of population experimented on($\mathbf{T}$) & 90\% & 50\% & 90\% & 90\% & 90\% & 50\% \\
Number of Buckets & \multicolumn{6}{c}{100}\\
Number of Days & \multicolumn{6}{c}{90}\\
Number of Starting points &  \multicolumn{6}{c}{50}\\
Number of Replications per starting point &  \multicolumn{6}{c}{10000}\\
\hline
\end{tabular}
\caption{Parameters of the Monte Carlo Simulation. For parameters with lists of values, a value is randomly (uniformly) drawn when a new experiment is started.}\label{tab:sim_corr2}
\end{table}
Figures \ref{fig:sim_corr12} and \ref{fig:sim_corr22} display the simulation results. The patterns are the same as in Section \ref{sec:sim_corr}, but the biases are generally larger for the same settings when the number of buckets is smaller. This is expected, as the heterogeneity is 'fixed' between these settings, in the sense that the ATE's for the first time points are drawn from the same normal distribution. This implies that there are going to be more and more buckets (when the number of buckets increase) with similar values, in turn implying that the heterogeneity between samples will decrease.

\begin{figure}[h!]\centering
\caption{The ATE1 bias, bucket availability correlation, and bucket sampling correlation plotted as a function of $\delta$, for settings 1-2. The parameters for each setting are given in Table \ref{tab:sim_corr}. }\label{fig:sim_corr12}
\includegraphics[scale=1.2, trim={5cm 16.1cm 3cm 4.5cm},page=11, clip=TRUE]{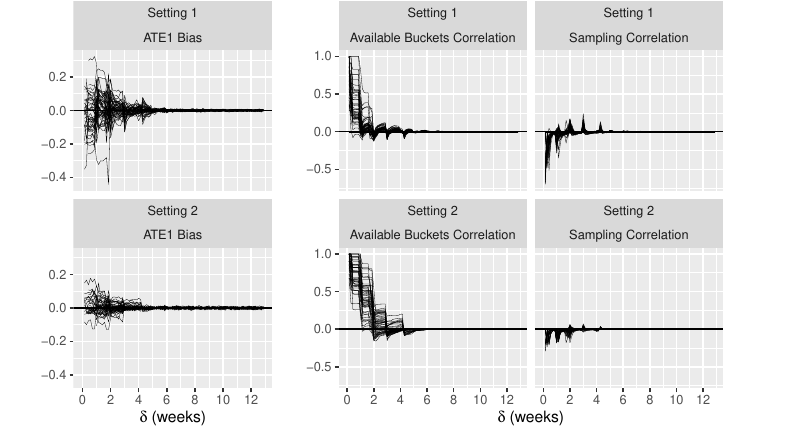}
\end{figure}
\begin{figure}[h!]\centering
\caption{The ATE1 bias, bucket availability correlation, and bucket sampling correlation plotted as a function of $\delta$, for settings 3-6. The parameters for each setting are given in Table \ref{tab:sim_corr}. }\label{fig:sim_corr22}
\includegraphics[scale=1.2, trim={5cm 8.5cm 3cm 4.5cm},page=12, clip=TRUE]{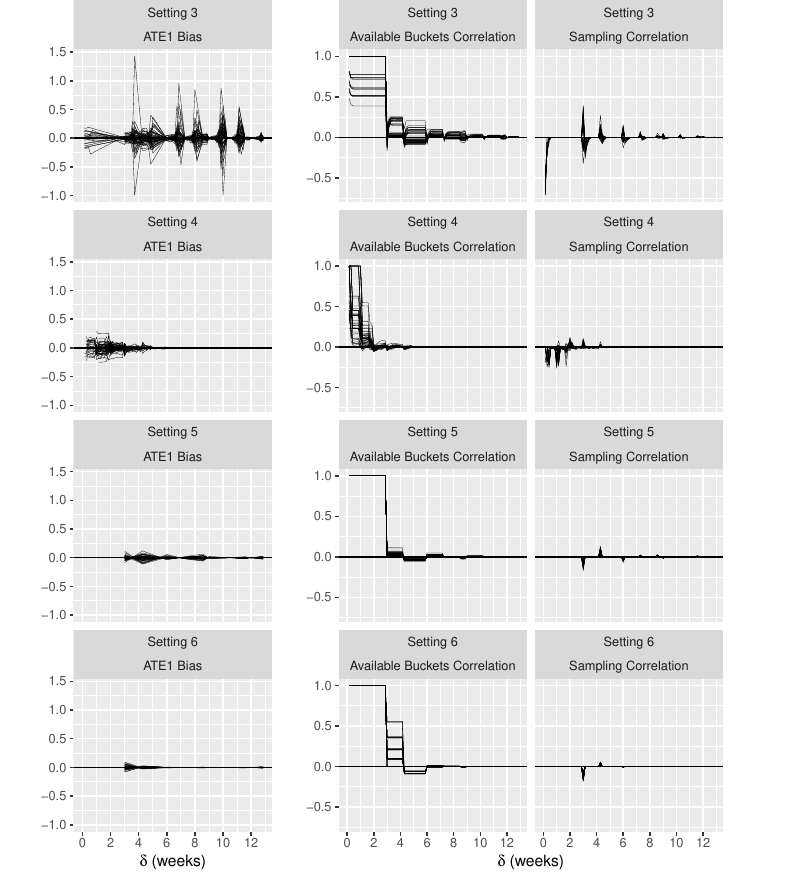}
\end{figure}

\end{document}